\documentclass[reprint, amsmath,amssymb,aps]{revtex4-1}

\usepackage[utf8]{inputenc}
\usepackage{circuitikz-0.8.3} 
\usepackage{graphicx}
\usepackage{dsfont}
\usepackage{qcircuit} 
\usepackage[caption=false]{subfig}
\usepackage{verbatim}
\usepackage{tikz}
\usepackage{amssymb,amsmath,amsfonts,amsthm}
\usepackage{xcolor}
\usepackage{physics}
\usepackage{bbm}
\usepackage{dcolumn}
\usepackage{geometry}
\newtheorem{theorem}{Theorem}

\newtheorem{remark}{Remark}

\newtheorem{lemma}{Lemma}
\newcommand{\I}{\mathbbm{1}} 
\newcommand{\CN}{{\sf{CN}}} 

\newcommand{\TG}{\textbf{T}}

\DeclareMathOperator*{\argmin}{\arg\!\min}
\begin{document}
\title{Abrupt Transitions in Variational Quantum Circuit Training}

\author{Ernesto Campos}%
 \email{ernesto.campos@skoltech.ru}
 \affiliation{
 Skolkovo Institute of Science and Technology, 3 Nobel Street, Moscow 143026, Russian Federation 
 }
\author{Aly Nasrallah}%
 \email{aly.nasrallah@skoltech.ru}
 \affiliation{
 Skolkovo Institute of Science and Technology, 3 Nobel Street, Moscow 143026, Russian Federation 
 }
\author{Jacob Biamonte}
\email{j.biamonte@skoltech.ru}
 \homepage{http://quantum.skoltech.ru}
 \affiliation{
 Skolkovo Institute of Science and Technology, 3 Nobel Street, Moscow 143026, Russian Federation
 }

\begin{abstract}
Variational quantum algorithms dominate gate-based applications of modern quantum processors.  The so called, {\it layer-wise trainability conjecture} appears in various works throughout the variational quantum computing literature.  The conjecture asserts that a quantum circuit can be trained piece-wise, e.g.~that a few layers can be trained in sequence to minimize an objective function.  Here we prove this conjecture false. Counterexamples are found by considering objective functions that are exponentially close (in the number of qubits) to the identity matrix.  In the finite setting, we found abrupt transitions in the ability of quantum circuits to be trained to minimize these objective functions.  Specifically, we found that below a critical (target gate dependent) threshold, circuit training terminates close to the identity and remains near to the identity for subsequently added blocks trained piece-wise.  A critical layer depth will abruptly train arbitrarily close to the target, thereby minimizing the objective function.  These findings shed new light on the divide-and-conquer trainability of variational quantum circuits and apply to a wide collection of contemporary literature.
\end{abstract}
\maketitle

Recent times have seen dramatic advancements in the experimental realization of gate-based quantum processors.  Tantamount results include the demonstration of certain sampling tasks beyond the reach of classical supercomputers \cite{arute2019quantum} as well as mid-sized prototype demonstrations of quantum approximate optimisation \cite{farhi2014quantum}, quantum chemistry \cite{mccaskey2019quantum,o2016scalable,peruzzo2014variational,yung2014transistor,hempel2018quantum,colless2018computation} and various machine learning tasks \cite{smith2019simulating,Ma2020,biamonte2017quantum,broughton2020tensorflow,grant2018hierarchical,benedetti2019generative}.  To circumvent and avoid physical limitations experienced in today's hardware, these demonstrations were enabled by the so called, variational model of quantum computation \cite{moll2018quantum,peruzzo2014variational,morales2019universality}. The variational model works in tandem with a classical co-processor: iterative adjustments tune a parameterized quantum circuit to minimize an objective function.

The compilation problem---a traditional problem dating from the start of the field of quantum information processing \cite{Nielsen2011}---is to determine a sequence of simplistic gates that realizes a more complicated larger operation \cite{kitaevclassquan}.  Today's quantum processors realize fixed configurations of tunable gates.  For example, the so called hardware efficient ansatz (HEA)~\cite{kandala2017hardware} stems from the physical fact that qubits must interact by means of neighboring pairs.  The HEA---as well as other so called ansatz circuits---arise from physical motivations and form a fixed structure of gates that can be trained to minimize objective functions as well as other tasks.  The modern method of variational quantum computation has already influenced the problem of compilation \cite{Khatri2019quantumassisted} by using fixed ansatz circuits to realize desired operations \cite{higgott2019variational, biamonte2019universal, kirby2019contextuality, mitarai2018quantum, hastings2019classical}.

Subsequent studies of ansatz based compilation have considered noise resilience \cite{Sharma_2020} and adapted these methods to various forms of gradient-based and gradient-free optimization. State of the art results can be found in e.g.~the works \cite{Khatri2019quantumassisted,mcclean2018barren,mitarai2018quantum,schuld2019evaluating,xu2019variational}.  Closely related to variational approaches to compilation are the variational class of quantum algorithms, which have proven effective at certain machine learning tasks \cite{liang2020variational,mcclean2016theory} in which the variational circuit becomes a generative machine learning model \cite{benedetti2019parameterized,verdon2017quantum}. Both approaches typically adopt a layer-wise divide and conquer training strategy.

The so called, piece-wise trainability conjecture appears throughout the variational quantum computing literature, importantly see \cite{Skolik2020,carolan2020variational}.  The conjecture asserts that a circuit can be trained piece-wise, e.g.~that a few layers can be trained to form the first block, composed with a further block of layers and all the while, an objective function can be iteratively minimized. This conjecture turns out to not always be true.  

In the finite setting, we found abrupt transitions in the ability for quantum circuits to be trained.  Specifically, we consider the HEA and the checkerboard ansatze. Analytically we determine that for both ansatze the identity is an extrema, which prevents further training. As a consequence, the heuristic that aims to avoid barren plateaus by initializing the gates as identities  \cite{Grant2019, Skolik2020} does not apply. 

For a number of layers per stack below a critical threshold, the circuit trains close to the identity for the case of the HEA and stops training after a few layers (usually the first) for the checkerboard ansatz.  Abruptly, a critical layer depth per stack will train arbitrarily close to the target for both ansatze.  Indeed, the objective functions we consider are not some abstract invention, for example, a $k$-controlled unitary operator satisfies our criteria and provides a class of counterexamples to the piece-wise trainability conjecture.

Our results illustrate an abrupt trainability transition at critical target gate dependent depths. We confirm this analytically for single layers.

\section{Methods}\label{methods}

\textbf{Variational state- and gate-learning}.
Several approaches exist to transform unitary matrices into sequences of gates and operations that can be realized on quantum processors.  Modern quantum processors suffer from many limitations \cite{butko2019understanding}.  The variational class of quantum algorithms seeks to circumvent some systematic limitations such as variability in pulse timing and limited coherence times. 
Variational quantum algorithms are a class of hybrid quantum-classical algorithms that involve the minimization of a cost function dependant on the parameters of a tunable quantum gate sequence $U(\boldsymbol{\theta})$. 

We seek to increase the Hilbert-Schmidt overlap between $U(\boldsymbol{\theta})$ and a target unitary $\TG$, as follows:
\begin{enumerate}
    \item Prepare a parametrized gate sequence of the form 
    \begin{equation}
        U(\boldsymbol{\theta})=\prod_{i=1}^{p}V^{(i)}(\boldsymbol{\theta}_i),
    \end{equation}
    where $V^{(i)}(\boldsymbol{\theta}_i)$ are the ansatz layers. This circuit will be used to approximate a target unitary $\TG$. 
    
    \item Calculate the cost function to train the parametrized gate sequence, viz.,
    \begin{equation*}
     d\big\{U(\boldsymbol{\theta})\in \mathcal{L}(\left[\mathbb{C}^2\right]^{\otimes n}),~\TG\in \mathcal{L}(\left[\mathbb{C}^2\right]^{\otimes n}) \big\}~ \in~ [0, \infty).
    \end{equation*}
    
    \item The algorithm then updates the parameters based on the calculation of the cost function. 
    
    \item Iterate steps 2 and 3 until a certain threshold is reached.
\end{enumerate}

\textbf{Training by layer stacks}.
In an attempt to avoid barren plateaus and reduce optimization time, commonly $U(\boldsymbol{\theta})$ is partitioned and trained piece-wise \cite{Skolik2020,carolan2020variational}.
A layer stack $S_j$ is a block composed of $j$ layers of an ansatz. After an initial stack is optimized, its parameters are fixed and the next stack of $j$ layers is considered
\begin{equation}
\begin{aligned}
 S_j^{(m)}(\boldsymbol{\theta}_m)&=\prod_{i}^{j}V^{(i)}(\boldsymbol{\theta}_i),\\
    U(\boldsymbol{\theta})&=\prod_{m=1}^q S_{j}^{(m)}(\boldsymbol{\theta}_m), 
\end{aligned}
\end{equation}
where $m$ corresponds to each stack being added, and the total depth of the circuit is $p=q\cdot j$.

\textbf{Hardware efficient ansatz}. A single layer of the hardware efficient ansatz consists of three single-qubit rotations on each wire followed by a control rotation in a daisy chain, as depicted in figure \ref{HEAfig}.
 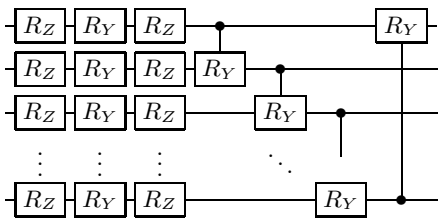
\begin{figure}[ht]
\centering
\[\Qcircuit @C=.4em @R= .4em @!R{
            & \gate{R_Z} & \gate{R_Y}  & \gate{R_Z} & \ctrl{1}   & \qw        & \qw & \gate{R_Y}    & \qw  \\
            & \gate{R_Z} & \gate{R_Y}  & \gate{R_Z} & \gate{R_Y}   & \ctrl{1}        & \qw & \qw    & \qw  \\
             & \gate{R_Z} & \gate{R_Y}  & \gate{R_Z} & \qw   & \gate{R_Y}   & \ctrl{1}  & \qw    & \qw  \\
             &\vdots &\vdots  & \vdots & &\ddots & & &\\
             & \gate{R_Z} & \gate{R_Y}  & \gate{R_Z} & \qw   & \qw      & \gate{R_Y}& \ctrl{-4}   & \qw \\
        }\]
        \caption{One layer of the hardware efficient ansatz.}
        \label{HEAfig}
\end{figure}

\textbf{Checkerboard ansatz.} The checkerboard ansatz (a.k.a. brick layer ansatz) \cite{uvarov2020} appears throughout the literature and is depicted in figure \ref{fig:checkerboardl}.

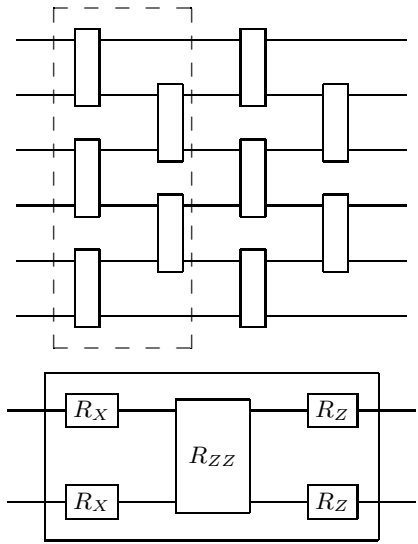
\begin{figure}[ht]
\[   \Qcircuit @C=2.35em @R= 1.35em {
 & \multigate{1}{} & \qw &\multigate{1}{} &\qw&\qw\\
 & \ghost{} & \multigate{1}{}& \ghost{}& \multigate{1}{}&\qw\\
 & \multigate{1}{} &\ghost{} &\multigate{1}{} &\ghost{}&\qw\\
 & \ghost{} & \multigate{1}{} &\ghost{}&\multigate{1}{}&\qw\\
 & \multigate{1}{} & \ghost{}& \multigate{1}{}&\ghost{}&\qw\\
 & \ghost{} &\qw &\ghost{}&\qw&\qw
 \gategroup{1}{2}{6}{3}{1.7em}{--}
}\]
\[   \Qcircuit @C=2.35em @R= 2.35em {
              & \gate{R_X} & \multigate{1}{R_{ZZ}}   & \gate{R_Z}   & \qw  \\
            & \gate{R_X} & \ghost{R_{ZZ}} &    \gate{R_Z} & \qw
             \gategroup{1}{4}{2}{2}{1.7em}{-}
}\]
    \caption{Top: 2 layers of the checkerboard ansatz. The dashed region encloses one ansatz layer. Bottom: the entangling gate used herein, where $R_{ZZ}(\omega)=e^{-\imath\omega Z\otimes Z}$.}
    \label{fig:checkerboardl}
\end{figure}

Without loss of generality, we chose to approximate $k$-Toffoli gates and vary $k$. 
The $k$-Toffoli gate can be written as 
\begin{equation}
X^{a_1...a_k}_b=\mathds{1}_{2^n-2}\oplus X,
\end{equation}
where the $\{a_i\}$ denotes control qubits and $b$ indexes the target qubit.
The cost function is based on the Hilbert-Schmidt matrix norm 
\begin{equation}
    d(X,Y)=1-\text{abs}(||X, Y||_{HS})\geq 0.
    \label{distance}
\end{equation}
This cost function is a metric as it satisfies positivity, symmetry, and the triangular inequality. 
If the training using $j$ layers per stack is not successful after $q$ stacks, $j$ is increased until the training is successful.  

The distance between a $k$-controlled gate and the identity is:
\begin{equation}
\begin{aligned}
d(A^{a_1...a_k}_b, \mathds{1}_{2^n})&=1-\frac{\Tr(A^{a_1...a_k\dagger}_b \mathds{1}_{2^n})}{2^n}\\
&=1-\frac{\Tr(\mathds{1}_{2^n-2}\oplus A)}{2^n}\\
&\leq2^{1-n}
=2^{-k}.
\end{aligned}
\label{deq}
\end{equation}
This will be used to infer from the cost when the optimization terminates near the identity. 

\section{Results}\label{results}
We make the following empirical observation: Given a metric function $d$, there exists a target unitary $\TG\neq e^{\imath\beta}\I$, $\beta \in [0, 2\pi)$, for which a stack $S_j(\boldsymbol{\theta})$, with depth $j<c$, added on top of a previously trained circuit $L\neq \TG$ will have a global minimum $S_j(\boldsymbol{\theta})=e^{\imath\alpha(\boldsymbol{\theta})}\I$, $\alpha(\boldsymbol{\theta}) \in [0, 2\pi)$. As a consequence, training by layer stacks with $<c$ layers is prevented, namely:
\begin{equation}
\begin{aligned}
   \forall ~\{S_j(\boldsymbol{\theta}),d\}& ~ \exists~ \{\TG\neq e^{\imath\beta}\I\}|\\
   &\argmin_{S_{j<c}(\boldsymbol{\theta})}~d(L\cdot S_{j<c}(\boldsymbol{\theta}),\TG)=e^{\imath\alpha(\boldsymbol{\theta})}\I.
 \end{aligned}
\end{equation} 
\newline Note that the cost function is gauge invariant under $U(1)$ placing an equivalence on all identity up to a phase.
\subsection{HEA results}
By training up to 5 qubits with up to 11 layers per stack of the hardware efficient ansatz we obtained the results shown in figure \ref{controledcomparison} and \ref{controlledcomparisonnorm}. We observe an abrupt trainability transition dependant on the target gate. We denote the critical number of layers per stack, $c$.

\begin{figure}[ht]
    \centering
    \includegraphics[width=0.5\textwidth]{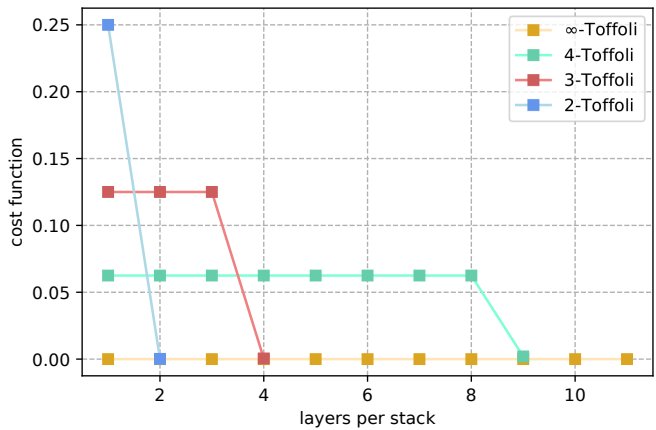}
    \caption{$k$-Toffoli gate trainability for multiple numbers of layers per stack of the HEA. E.g, the green plot (corresponding to the 4-Toffoli) illustrates untrainability for up to 8 layers per stack.}
    \label{controledcomparison}
\end{figure}
\begin{figure}[ht]
    \centering
    \includegraphics[width=0.5\textwidth]{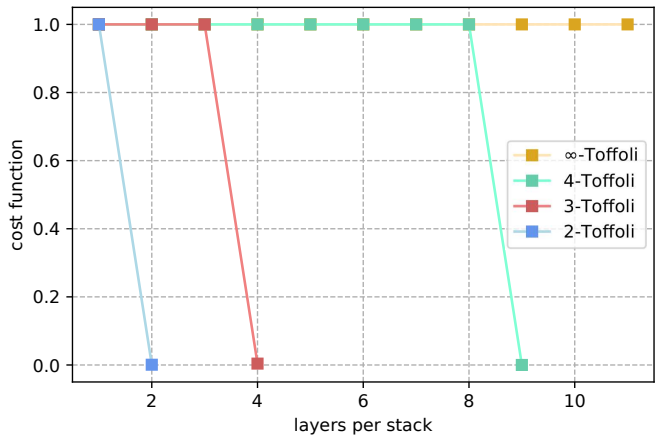}
    \caption{Trainability transition with the normalized value of the cost function for stacks of the HEA. We observe a waterfall effect on the critical number of layers per stack that increases with $k=n-1$.}
    \label{controlledcomparisonnorm}
\end{figure}

Bellow $c$ the cost is always $d(k\text{-Toffoli},\I)=2^{-k}$. For the 2-Toffoli gate, the global minimum of a single layer is the identity as depicted in figure \ref{1layerdists}. 

Furthermore, for a single layer of the HEA, we prove that $d( V_{\text{hea}}(\boldsymbol{\theta}),k\text{-Toffoli})$ takes extrema for the roots of
\begin{equation}
   V_{\text{hea}}(\boldsymbol{\theta}_\I)-e^{\imath\alpha(\boldsymbol{\theta}_\I)}\I=0.
   \label{identcondihea}
\end{equation}
We derived the complete set of hyperparameters for which \eqref{identcondihea} is true, and computed the first and second derivatives of $d$  (derivation in appendix \ref{min-cost}).

The parameters of the $R_Z(\theta_i)$ pair in each wire are not independent, and the gradients of $d$ with respect to them are 0 for the first and second derivatives, implying not just an extreme point but a valley. 
The gradients of $d$ with respect to the $CR_Y(\theta_i)$ and $R_Y(\theta_i)$ gate parameters are 0 for the first derivative and take a positive value for the second derivative. 

Since $V_{\text{hea}}(\boldsymbol{\theta}_\I)=e^{\imath\alpha(\boldsymbol{\theta}_\I)}\I$ are extrema, initializing the gates as identities (a heuristic suggested and used in \cite{Skolik2020, Grant2019}) would prevent any training.

\begin{figure}
    \includegraphics[width=.5\textwidth]{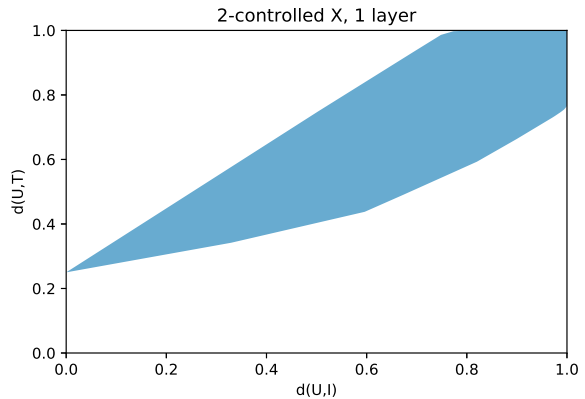}
    \caption{Contour of 10 million points representing unitaries made by a single layer of the HEA with random parameters, and their distance to the identity and the Toffoli gate illustrating $e^{\imath\alpha(\boldsymbol{\theta}_\I)} \I$ is the global minimum.}
    \label{1layerdists}
\end{figure}

When training stacks with more than $c$ layers, it is still common to find the optimization returning an identity. Avoiding the identity extrema becomes increasingly difficult as the number of qubits increases. This stems from the distance between the identity and the $k$-Toffoli decreasing exponentially with $k$ as $d=2^{-k}$.

\subsection{Checkerboard ansatz results}
Abrupt trainability transitions appear in checkerboard in a slightly different fashion compared to the HEA. 
Figure \ref{fig:2toffcheck} illustrates the training process for a 2-Toffoli. When training stacks of 1 and 2 layers, training stops after a single stack, as the following stacks approximate an identity. This implies the identity is not the global minimum for the first layer stack (as it is the case for HEA), yet it becomes the global minimum following a stack that trained.
It is only after $c\geq3$ we observe training with each new stack. 

\begin{figure}[ht]
    \centering
    \includegraphics[width=0.5\textwidth]{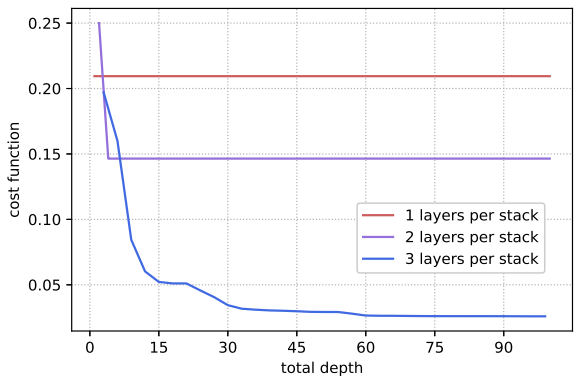}
    \caption{Trainability convergence for 1,2 and 3 layer stacks of the checkerboard ansatz approximating a 2-Toffoli gate. 1 and 2 layer stacks stop training after one stack trains, but 3 improves the approximation with every added stack.}
    \label{fig:2toffcheck}
\end{figure}
Similar to the HEA, for a single layer of the checkerboard ansatz, $d(V_\text{ch}(\boldsymbol{\omega}),k\text{-Toffoli})$ takes extrema for the roots of 
\begin{equation}
     V_\text{ch}(\boldsymbol{\omega}_\I)-e^{\imath\alpha(\boldsymbol{\omega}_\I)}\I=0.
\end{equation}
Computing the gradients of $d$  with respect to any parameter, the first derivative of  $d$ is 0, and the second derivative is a positive number (derivation in appendix \ref{min-cost}).

\section{Conclusion}\label{conclusion}

We have proven the layer-wise trainability conjecture false by showing the existence of {\it abrupt trainability transitions}. For the cases studied, training layer-wise with less than a critical number of layers per stack prevented learning, as each stack would train close to an identity. The effect was observed in the HEA and checkerboard ansatze when trying to approximate gates with $O(2^{-n})$ distance to the identity.
The untrainability seems to affect ansatze that can approximate an identity up to a global phase, but not a target unitary that has a high Hilbert-Schmidt overlap with the identity.

The findings depart from known results: various papers of the area studied similar scenarios but did not encounter abrupt transitions. For example, Skolik et al.~\cite{Skolik2020} studied training by layer stacks as a strategy to avoid barren plateaus. The researchers did not encounter abrupt transitions as their problem (MNIST classification) does not have the identity extrema in the cost function. The paper by Sharma et al.~\cite{Sharma_2020} studied quantum compilers noise resilience. The authors did not observe an abrupt transition when training a Toffoli gate as the ansatz they used (checkerboard composed of dressed CNOTs) can not simulate an identity for a single layer. 

Also, our results question the benefits of using an initialization heuristic where every gate begins as an identity (as used in \cite{Skolik2020}) or where randomly selected gates are initialized as identities \cite{Grant2019}, in an attempt to avoid barren plateaus. We prove such initialization would prevent training for the cases we studied.

\begin{acknowledgments}
The  authors  acknowledge  support  from  the  project, {\it Leading  Research  Center  on  Quantum  Computing} (Agreement No. 014/20). 
\end{acknowledgments}

\onecolumngrid
\bibliography{refs}

\begin{thebibliography}{10}

\bibitem{arute2019quantum}
Frank Arute, Kunal Arya, Ryan Babbush, Dave Bacon, Joseph~C Bardin, Rami
  Barends, Rupak Biswas, Sergio Boixo, Fernando~GSL Brandao, David~A Buell,
  et~al.
\newblock Quantum supremacy using a programmable superconducting processor.
\newblock {\em Nature}, 574(7779):505--510, 2019.

\bibitem{farhi2014quantum}
Edward Farhi, Jeffrey Goldstone, and Sam Gutmann.
\newblock A quantum approximate optimization algorithm.
\newblock {\em arXiv preprint arXiv:1411.4028}, 2014.

\bibitem{mccaskey2019quantum}
Alexander~J McCaskey, Zachary~P Parks, Jacek Jakowski, Shirley~V Moore, Titus~D
  Morris, Travis~S Humble, and Raphael~C Pooser.
\newblock Quantum chemistry as a benchmark for near-term quantum computers.
\newblock {\em npj Quantum Information}, 5(1):1--8, 2019.

\bibitem{o2016scalable}
Peter~JJ O’Malley, Ryan Babbush, Ian~D Kivlichan, Jonathan Romero, Jarrod~R
  McClean, Rami Barends, Julian Kelly, Pedram Roushan, Andrew Tranter, Nan
  Ding, et~al.
\newblock Scalable quantum simulation of molecular energies.
\newblock {\em Physical Review X}, 6(3):031007, 2016.

\bibitem{peruzzo2014variational}
Alberto Peruzzo, Jarrod McClean, Peter Shadbolt, Man-Hong Yung, Xiao-Qi Zhou,
  Peter~J Love, Al{\'a}n Aspuru-Guzik, and Jeremy~L O’brien.
\newblock A variational eigenvalue solver on a photonic quantum processor.
\newblock {\em Nature communications}, 5:4213, 2014.

\bibitem{yung2014transistor}
M-H Yung, Jorge Casanova, Antonio Mezzacapo, Jarrod Mcclean, Lucas Lamata, Alan
  Aspuru-Guzik, and Enrique Solano.
\newblock From transistor to trapped-ion computers for quantum chemistry.
\newblock {\em Scientific reports}, 4:3589, 2014.

\bibitem{hempel2018quantum}
Cornelius Hempel, Christine Maier, Jonathan Romero, Jarrod McClean, Thomas
  Monz, Heng Shen, Petar Jurcevic, Ben~P Lanyon, Peter Love, Ryan Babbush,
  et~al.
\newblock Quantum chemistry calculations on a trapped-ion quantum simulator.
\newblock {\em Physical Review X}, 8(3):031022, 2018.

\bibitem{colless2018computation}
James~I Colless, Vinay~V Ramasesh, Dar Dahlen, Machiel~S Blok,
  ME~Kimchi-Schwartz, JR~McClean, J~Carter, WA~De~Jong, and I~Siddiqi.
\newblock Computation of molecular spectra on a quantum processor with an
  error-resilient algorithm.
\newblock {\em Physical Review X}, 8(1):011021, 2018.

\bibitem{smith2019simulating}
Adam Smith, MS~Kim, Frank Pollmann, and Johannes Knolle.
\newblock Simulating quantum many-body dynamics on a current digital quantum
  computer.
\newblock {\em npj Quantum Information}, 5(1):1--13, 2019.

\bibitem{Ma2020}
He~Ma, Marco Govoni, and Giulia Galli.
\newblock {Quantum simulations of materials on near-term quantum computers}.
\newblock {\em npj Computational Materials 2020 6:1}, 6(1):1--8, feb 2020.

\bibitem{biamonte2017quantum}
Jacob Biamonte, Peter Wittek, Nicola Pancotti, Patrick Rebentrost, Nathan
  Wiebe, and Seth Lloyd.
\newblock Quantum machine learning.
\newblock {\em Nature}, 549(7671):195--202, 2017.

\bibitem{broughton2020tensorflow}
Michael Broughton, Guillaume Verdon, Trevor McCourt, Antonio~J Martinez,
  Jae~Hyeon Yoo, Sergei~V Isakov, Philip Massey, Murphy~Yuezhen Niu, Ramin
  Halavati, Evan Peters, et~al.
\newblock Tensorflow quantum: A software framework for quantum machine
  learning.
\newblock {\em arXiv preprint arXiv:2003.02989}, 2020.

\bibitem{grant2018hierarchical}
Edward Grant, Marcello Benedetti, Shuxiang Cao, Andrew Hallam, Joshua Lockhart,
  Vid Stojevic, Andrew~G Green, and Simone Severini.
\newblock Hierarchical quantum classifiers.
\newblock {\em npj Quantum Information}, 4(1):1--8, 2018.

\bibitem{benedetti2019generative}
Marcello Benedetti, Delfina Garcia-Pintos, Oscar Perdomo, Vicente
  Leyton-Ortega, Yunseong Nam, and Alejandro Perdomo-Ortiz.
\newblock A generative modeling approach for benchmarking and training shallow
  quantum circuits.
\newblock {\em npj Quantum Information}, 5(1):1--9, 2019.

\bibitem{moll2018quantum}
Nikolaj Moll, Panagiotis Barkoutsos, Lev~S Bishop, Jerry~M Chow, Andrew Cross,
  Daniel~J Egger, Stefan Filipp, Andreas Fuhrer, Jay~M Gambetta, Marc Ganzhorn,
  et~al.
\newblock Quantum optimization using variational algorithms on near-term
  quantum devices.
\newblock {\em Quantum Science and Technology}, 3(3):030503, 2018.

\bibitem{morales2019universality}
Mauro~ES Morales, Jacob Biamonte, and Zolt{\'a}n Zimbor{\'a}s.
\newblock On the universality of the quantum approximate optimization
  algorithm.
\newblock {\em arXiv preprint arXiv:1909.03123}, 2019.

\bibitem{Nielsen2011}
Michael~A. Nielsen and Isaac~L. Chuang.
\newblock {\em Quantum Computation and Quantum Information: 10th Anniversary
  Edition}.
\newblock Cambridge University Press, USA, 10th edition, 2011.

\bibitem{kitaevclassquan}
A.~Yu. Kitaev, A.~H. Shen, and M.~N. Vyalyi.
\newblock {\em Classical and Quantum Computation}.
\newblock American Mathematical Society, USA, 2002.

\bibitem{kandala2017hardware}
Abhinav Kandala, Antonio Mezzacapo, Kristan Temme, Maika Takita, Markus Brink,
  Jerry~M Chow, and Jay~M Gambetta.
\newblock Hardware-efficient variational quantum eigensolver for small
  molecules and quantum magnets.
\newblock {\em Nature}, 549(7671):242--246, 2017.

\bibitem{Khatri2019quantumassisted}
Sumeet Khatri, Ryan LaRose, Alexander Poremba, Lukasz Cincio, Andrew~T.
  Sornborger, and Patrick~J. Coles.
\newblock Quantum-assisted quantum compiling.
\newblock {\em {Quantum}}, 3:140, May 2019.

\bibitem{higgott2019variational}
Oscar Higgott, Daochen Wang, and Stephen Brierley.
\newblock Variational quantum computation of excited states.
\newblock {\em Quantum}, 3:156, 2019.

\bibitem{biamonte2019universal}
Jacob Biamonte.
\newblock Universal variational quantum computation.
\newblock {\em arXiv preprint arXiv:1903.04500}, 2019.

\bibitem{kirby2019contextuality}
William~M Kirby and Peter~J Love.
\newblock Contextuality test of the nonclassicality of variational quantum
  eigensolvers.
\newblock {\em Physical Review Letters}, 123(20):200501, 2019.

\bibitem{mitarai2018quantum}
Kosuke Mitarai, Makoto Negoro, Masahiro Kitagawa, and Keisuke Fujii.
\newblock Quantum circuit learning.
\newblock {\em Physical Review A}, 98(3):032309, 2018.

\bibitem{hastings2019classical}
Matthew~B Hastings.
\newblock Classical and quantum bounded depth approximation algorithms.
\newblock {\em arXiv preprint arXiv:1905.07047}, 2019.

\bibitem{Sharma_2020}
Kunal Sharma, Sumeet Khatri, M~Cerezo, and Patrick~J Coles.
\newblock Noise resilience of variational quantum compiling.
\newblock {\em New Journal of Physics}, 22(4):043006, apr 2020.

\bibitem{mcclean2018barren}
Jarrod~R McClean, Sergio Boixo, Vadim~N Smelyanskiy, Ryan Babbush, and Hartmut
  Neven.
\newblock Barren plateaus in quantum neural network training landscapes.
\newblock {\em Nature Communications}, 9(1):1--6, 2018.

\bibitem{schuld2019evaluating}
Maria Schuld, Ville Bergholm, Christian Gogolin, Josh Izaac, and Nathan
  Killoran.
\newblock Evaluating analytic gradients on quantum hardware.
\newblock {\em Physical Review A}, 99(3):032331, 2019.

\bibitem{xu2019variational}
Xiaosi Xu, Simon~C Benjamin, and Xiao Yuan.
\newblock Variational circuit compiler for quantum error correction.
\newblock {\em arXiv preprint arXiv:1911.05759}, 2019.

\bibitem{liang2020variational}
Jin-Min Liang, Shu-Qian Shen, Ming Li, and Lei Li.
\newblock Variational quantum algorithms for dimensionality reduction and
  classification.
\newblock {\em Physical Review A}, 101(3):032323, 2020.

\bibitem{mcclean2016theory}
Jarrod~R McClean, Jonathan Romero, Ryan Babbush, and Al{\'a}n Aspuru-Guzik.
\newblock The theory of variational hybrid quantum-classical algorithms.
\newblock {\em New Journal of Physics}, 18(2):023023, 2016.

\bibitem{benedetti2019parameterized}
Marcello Benedetti, Erika Lloyd, Stefan Sack, and Mattia Fiorentini.
\newblock Parameterized quantum circuits as machine learning models.
\newblock {\em Quantum Science and Technology}, 4(4):043001, 2019.

\bibitem{verdon2017quantum}
Guillaume Verdon, Michael Broughton, and Jacob Biamonte.
\newblock A quantum algorithm to train neural networks using low-depth
  circuits.
\newblock {\em arXiv preprint arXiv:1712.05304}, 2017.

\bibitem{Skolik2020}
Andrea Skolik, Jarrod McClean, Masoud Mohseni, Patrick van~der Smagt, and
  Martin Leib.
\newblock {Layerwise learning for quantum neural networks}.
\newblock {\em Bulletin of the American Physical Society}, Volume 65, Number 1,
  2020.

\bibitem{carolan2020variational}
Jacques Carolan, Masoud Mohseni, Jonathan~P Olson, Mihika Prabhu, Changchen
  Chen, Darius Bunandar, Murphy~Yuezhen Niu, Nicholas~C Harris, Franco~NC Wong,
  Michael Hochberg, et~al.
\newblock Variational quantum unsampling on a quantum photonic processor.
\newblock {\em Nature Physics}, 16(3):322--327, 2020.

\bibitem{Grant2019}
Edward Grant, Leonard Wossnig, Mateusz Ostaszewski, and Marcello Benedetti.
\newblock {An initialization strategy for addressing barren plateaus in
  parametrized quantum circuits}.
\newblock {\em Quantum}, 3:214, dec 2019.

\bibitem{butko2019understanding}
Anastasiia Butko, George Michelogiannakis, Samuel Williams, Costin Iancu, David
  Donofrio, John Shalf, Jonathan Carter, and Irfan Siddiqi.
\newblock Understanding quantum control processor capabilities and limitations
  through circuit characterization.
\newblock {\em arXiv preprint arXiv:1909.11719}, 2019.

\bibitem{uvarov2020}
A.~V. Uvarov, A.~S. Kardashin, and J.~D. Biamonte.
\newblock Machine learning phase transitions with a quantum processor.
\newblock {\em Phys. Rev. A}, 102:012415, Jul 2020.

\end{thebibliography}
\bibliographystyle{unsrt}

\clearpage
\onecolumngrid
\begin{center}
\textbf{\large Supplementary Materials}
\end{center}
\appendix
\section{Cost function derivatives}\label{gradient analysis}

\noindent The cost function can be expressed as follows
\begin{equation}
 d(X^{a_1...a_k}_b,U(\boldsymbol{\theta}))=1-\frac{1}{N}\sqrt{\Trace(X^{a_1...a_k\dagger}_b U(\boldsymbol{\theta}))^*\Trace(X^{a_1...a_k\dagger}_b  U(\boldsymbol{\theta})) }~.
 \label{distamceexplicit}
\end{equation}
The first derivative of the cost function is
\begin{equation}
\begin{aligned}
\partial_{j}d(X^{a_1...a_k}_b,U(\boldsymbol{\theta}))=& -{\frac{\Trace(X^{a_1...a_k\dagger}_b  ~\partial_{j}U(\boldsymbol{\theta}))^*\Trace(X^{a_1...a_k\dagger}_b  U(\boldsymbol{\theta}))}{ 2N\sqrt{\Trace(X^{a_1...a_k\dagger}_b U(\boldsymbol{\theta}))^*\Trace(X^{a_1...a_k\dagger}_b  U(\boldsymbol{\theta})) }}}\\
&-{\frac{\Trace(X^{a_1...a_k\dagger}_b  U(\boldsymbol{\theta}))^*\Trace(X^{a_1...a_k\dagger}_b  ~\partial_{j}U(\boldsymbol{\theta}))}{ 2N\sqrt{\Trace(X^{a_1...a_k\dagger}_b U(\boldsymbol{\theta}))^*\Trace(X^{a_1...a_k\dagger}_b  U(\boldsymbol{\theta}))}}}\\
=&-\frac{\Re\left\{\Trace(X^{a_1...a_k\dagger}_b  ~\partial_{j}U(\boldsymbol{\theta}))^*\Trace(X^{a_1...a_k\dagger}_b  U(\boldsymbol{\theta}))\right\}}{ {N\sqrt{\Trace(X^{a_1...a_k\dagger}_b U(\boldsymbol{\theta}))^*\Trace(X^{a_1...a_k\dagger}_b  U(\boldsymbol{\theta})) } }}\label{distance_derivative},
\end{aligned}
\end{equation}
where $\partial_j=\frac{\partial}{\partial \theta_j}$. The second derivative is
\begin{equation}\label{distance_second_derivative}
\begin{aligned}
\partial_{j}^2d(X^{a_1...a_k}_b,U(\boldsymbol{\theta}))=&\frac{-\Re\left\{\Trace(X^{a_1...a_k\dagger}_b~\partial_{j}^2U(\boldsymbol{\theta}))^*\Trace(X^{a_1...a_k\dagger}_b  U(\boldsymbol{\theta}))\right\}-|\Trace(X^{a_1...a_k\dagger}_b  ~\partial_{j}U(\boldsymbol{\theta}))|^2}{N\left[\Trace(X^{a_1...a_k\dagger}_b U(\boldsymbol{\theta}))^*\Trace(X^{a_1...a_k\dagger}_b  U(\boldsymbol{\theta})) \right]^{1/2}}\\
&+ \frac{|\Trace(X^{a_1...a_k\dagger}_b \partial_{j} U(\boldsymbol{\theta}))\Trace(X^{a_1...a_k\dagger}_b  U(\boldsymbol{\theta}))|^2}{N\left[\Trace(X^{a_1...a_k\dagger}_b U(\boldsymbol{\theta}))^*\Trace(X^{a_1...a_k\dagger}_b U(\boldsymbol{\theta}))\right]^{3/2}}.
\end{aligned}
\end{equation}

\noindent The partial derivative of $U(\boldsymbol{\theta})$ with respect to a single qubit parameter $\boldsymbol{\theta}_{j}$ can be calculated as follows
\begin{equation}
\begin{aligned}
U(\boldsymbol{\theta})=&\cdots e^{-\imath\sigma\theta_{j}}\cdots\\
 =&\cdots R\sigma(\theta_{j})
 \cdots, \\
\partial_{j} U(\boldsymbol{\theta})=&\cdots R\sigma(\theta_j)(-\imath\sigma)
\cdots,
\end{aligned}
\end{equation}
where $\sigma$ is a Pauli matrix, and $\cdots$ represents the gates to the left and right of $R\sigma(\theta_j)$. In case there are no gates, $\cdots$ represents an identity.
\section{Extrema of the cost function}\label{min-cost}

\begin{remark}For the HEA, we were not able to evaluate the cost function derivatives \eqref{distance_derivative} solely by using equation \eqref{heaidentitycondition}, which prompted us to calculate the set of hyperparameters that satisfy \eqref{heaidentitycondition}.
\begin{equation}
    V_{\text{hea}}(\boldsymbol{\theta}_\I)-e^{\imath\alpha(\boldsymbol{\theta}_\I)}\I=0.
    \label{heaidentitycondition}
\end{equation}
\end{remark}

\begin{lemma} The structure of the HEA, depicted in figure \ref{HEAfig}, constrains the set of hyperparameters $\Theta$ that satisfy \eqref{heaidentitycondition} as:
\begin{equation}\label{heavarietyeqs}
\Theta= \left\{\begin{aligned}\boldsymbol{\theta}_{\I}~|~ \text{if}~~{\theta}_\text{Di+1}&=2 \pi w ~ \implies
       {\theta}_\text{Yi}=\pi u,~
    {\theta}_\text{Z1i}+{\theta}_\text{Z2i}=\pi v~
        \\
       \boldsymbol{\theta}_{\I}~|~ \text{if}~~{\theta}_\text{Di+1}&= \pi(2 w-1) ~\implies {\theta}_\text{Yi}=\pi u,~
    {\theta}_\text{Z1i}+{\theta}_\text{Z2i}=\pi\left(\frac{1}{2}+ v\right)~
           \end{aligned}  \right\}
           \begin{aligned}
                (a)\\
                (b)
           \end{aligned}
\end{equation}
where $u,v,w \in \mathbb{Z}$, $i$ is the qubit index, ${\theta}_{\text{Y}i}$ are the $R_Y$ parameters, ${\theta}_{\text{Z1}i}$ and ${\theta}_{\text{Z2}i}$ are the parameters of each pair of $R_Z$ gates in a wire, and  ${\theta}_{\text{D}i}$ are the $CR_Y$ parameters.
\end{lemma}

\begin{proof}A single layer of the HEA can be written as:
\begin{equation}
    V_\text{hea}(\boldsymbol{\theta})=D(\boldsymbol{\theta}_\text{D})R(\boldsymbol{\theta}_\text{R}),
\end{equation}
where $R(\boldsymbol{\theta}_\text{R})$ are the single qubit gates, and $D(\boldsymbol{\theta}_\text{D})$ is the $CR_Y$ gates daisy-chain. We want $V_\text{hea}(\boldsymbol{\theta})$ to be equal to the identity up to a phase
\begin{equation}
    \begin{aligned}
        D(\boldsymbol{\theta}_\text{D})R(\boldsymbol{\theta}_\text{R})\simeq &\I\\
        D(\boldsymbol{\theta}_\text{D})\simeq & R(\boldsymbol{\theta}_\text{R})^\dagger.\\
    \end{aligned}
\end{equation}
Since $D(\boldsymbol{\theta}_\text{D})$ contains only non-local gates, it can only be expressed as the tensor product of single qubit gates if and only if $\text{rank}(D(\boldsymbol{\theta}_\text{D}))=1$
\begin{equation}
       D(\boldsymbol{\theta}_\text{D})\simeq \bigotimes_{i=1}^n R_i(\boldsymbol{\theta}_{\text{R}i})^\dagger\iff \text{rank}(D(\boldsymbol{\theta}_\text{D}))=1,\label{Drestriction}
\end{equation}
where $R_i(\boldsymbol{\theta}_{\text{R}i})$ are the single qubit gates applied to the $i$th qubit. \eqref{Drestriction} also implies that gates in different qubits commute as
\begin{equation}
    \text{rank}(D(\boldsymbol{\theta}_\text{D}))=1\implies \left[CR_{Y}^{i,j},CR_{Y}^{k,l}\right]=0,
\end{equation}
where $i,j,k,l$ are qubit indices.
$CR_{Y}^{i,j}$ can be written as
\begin{equation}
    CR_{Y}^{i,j}=\I+\ket{1}\bra{1}^i\otimes(R_Y(\theta_j)-\I )^j,
\end{equation}
so the commutation relation is
\begin{equation}
\begin{aligned}
     \left[CR_{Y}^{i,j},CR_{Y}^{k,l}\right]=&\left(\ket{1}\bra{1}^i\otimes(R_Y(\theta_j)-\I )^j\right)\left(\ket{1}\bra{1}^k\otimes(R_Y(\theta_l)-\I )^l\right)\\
    &-\left(\ket{1}\bra{1}^k\otimes(R_Y(\theta_l)-\I )^l\right)\left(\ket{1}\bra{1}^i\otimes(R_Y(\theta_j)-\I )^j\right)=0.
    \label{commutation}
\end{aligned}
\end{equation}
In our case we use a daisy chain structure, so we are interested in $i\neq k,~j\neq l$. 
For equivalence \eqref{commutation} to always be true, the options are:
\begin{itemize}
    \item[(i)] at least one parameter is $2\pi w$, or
\item[(ii)] both parameters are $\theta_{j,l}= \pi(2w-1)$,
\end{itemize}
for $w\in\mathbb{Z}$. Writing $CR_{Y}^{i,j}$ as 
\begin{equation}
    CR_Y(\theta_i)=\frac{1}{2}\big(Z\otimes(\I-R_Y(\theta_i))-\I\otimes(R_Y(\theta_i+I))\big),
    \label{rankIZ}
\end{equation}
we can check that using restrictions (i) and (ii),  $\text{rank}(CR_Y(\pi w))=1$. 
Restriction (i) allows $D(\boldsymbol{\theta}_\text{D})$ to have a parameter not equal to a multiple of $\pi$ and still have commuting gates
\begin{equation}
    D(\boldsymbol{\theta}_\text{D})=A\circ CR_Y^{i,j}(\theta_{\text{D}i})\circ B,
    \label{Dcomutes}
\end{equation}
where $CR_Y^{i,j}(\theta_{\text{D}i})$ is the gate with parameter $\theta_j\neq \pi w$, and  $A$, $B$ repesent the gates to the left and right, and $\text{rank}(A), \text{rank}(B)=1$.
From \eqref{rankIZ} and \eqref{Dcomutes}, $\theta_{\text{D}j}=\pi w$ is required for  $\text{rank}(D(\boldsymbol{\theta}_\text{D}))=1$, which implies
\begin{equation}
    D(\boldsymbol{\theta}_\text{D})\simeq\bigotimes_{i=1}^n R_i(\theta_{\text{R}i})\iff \theta_{\text{D}i}=\pi w.
\end{equation}

For $\theta_{\text{D}i+1}=2\pi w$, $R_i(\boldsymbol{\theta}_{\text{R}i})\simeq \I$ and writing $R_i(\boldsymbol{\theta}_{\text{R}i})$ as its components
\begin{equation}
\begin{aligned}
       R_{Z}({\theta}_{\text{Z2}i})R_{Y}({\theta}_{\text{Y}i})R_{Z}({\theta}_{\text{Z1}i})=&e^{\imath\alpha(\boldsymbol{\theta}_\I)} \I,
       \end{aligned}
\end{equation}
implies the restrictions ${\theta}_{\text{Y}i}=\pi u$, and ${\theta}_{\text{Z1}i}+{\theta}_{\text{Z2}i}=\pi v$, for $u,v\in \mathbb{Z}$.
For $\theta_{\text{D}i+1}=\pi(2 w-1)$ and writing $R_i(\boldsymbol{\theta}_{\text{R}i})$ as its components
\begin{equation}
\begin{aligned}
       ZR_{Z}({\theta}_{\text{Z2}i})R_{Y}({\theta}_{\text{Y}i})R_{Z}({\theta}_{\text{Z1}i})=&e^{\imath\alpha(\boldsymbol{\theta}_\I)} \I,
       \end{aligned}
\end{equation}
implies the restrictions ${\theta}_{\text{Y}i}=\pi u$, and ${\theta}_{\text{Z1}i}+{\theta}_{\text{Z2}i}=\pi\left(\frac{1}{2}+ v\right)$, for $u,v\in \mathbb{Z}$.
 \end{proof}
  \begin{figure}[ht]
     \[   \Qcircuit @C=1.35em @R= 1.35em {
              & \gate{R_Z(\theta_1)} & \gate{R_Y(\theta_2)}  & \gate{R_Z(\theta_3)} & \ctrl{1}   & \qw        & \qw & \gate{R_Y(\theta_{4})}    & \qw & \qw        & \qw & \qw \\
            & \gate{R_Z(\theta_5)} & \gate{R_Y(\theta_6)}  & \gate{R_Z(\theta_7)} & \gate{R_Y(\theta_{8})}   & \ctrl{1}        & \qw & \qw    & \qw & \qw  & \qw & \qw \\
             & \gate{R_Z(\theta_9)} & \gate{R_Y(\theta_{10})}  & \gate{R_Z(\theta_{11})} & \qw   & \gate{R_Y(\theta_{12})}        & \ctrl{1} & \qw    & \qw & \qw        & \qw & \qw \\
             &\vdots & \vdots & \vdots & & \ddots  & & & & & &\\
                 & \gate{R_Z(\theta_{m-3})} & \gate{R_Y(\theta_{m-2})}  & \gate{R_Z(\theta_{m-1})} & \qw   & \qw        & \gate{R_Y(\theta_{m})} & \ctrl{-4}    & \qw & \qw        & \qw & \qw 
             \gategroup{1}{1}{5}{4}{1.7em}{--}
        }\]
        \caption{Single layer of the HEA. Inside the dashed area $R(\boldsymbol{\theta}_\text{R})$, outside $D(\boldsymbol{\theta}_\text{D})$.}
        \label{Heaselected}
    \end{figure}
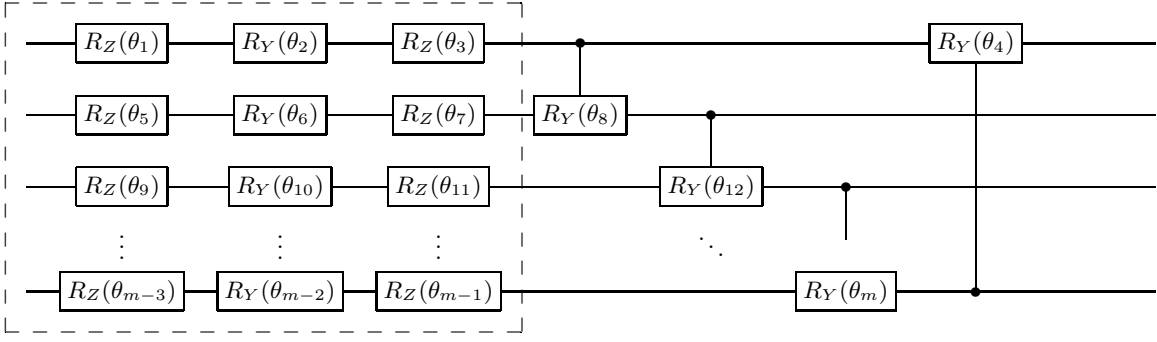

\begin{theorem}[Extrema of the HEA] For a single layer of the HEA, $d$ takes extrema  for the roots of 
\begin{equation}
    V_{\text{hea}}(\boldsymbol{\theta}_\I)-e^{\imath\alpha(\boldsymbol{\theta}_\I)}\I=0.
\end{equation}

\end{theorem}

\begin{proof} 
We prove $V_{\text{hea}}(\boldsymbol{\theta}_\I)$ are extrema by showing 
\begin{equation}
\begin{aligned}
    \partial_j d(V_{\text{hea}}(\boldsymbol{\theta}),k\text{-Toffoli})|_{\boldsymbol{\theta}=\boldsymbol{\theta}_\I}&=0,\\
    \partial^2_j d(V_{\text{hea}}(\boldsymbol{\theta}),k\text{-Toffoli})|_{\boldsymbol{\theta}=\boldsymbol{\theta}_\I}&\geq0.
\end{aligned}
\end{equation}

First, we compute the derivatives for the first set of parameter constrains \eqref{heavarietyeqs}--(a).
The partial derivative of $V_{\text{hea}}(\boldsymbol{\theta})$ with respect to a pair of $R_Z$ parameters is
\begin{align}
 \partial_{\text{Z1}i} V_{\text{hea}}(\boldsymbol{\theta})=\cdots((-\imath Z) R_{Z}(\theta_{\text{Z1}i})R_{Y}(\theta_{\text{Y}i})R_{Z}(\pi v-\theta_{\text{Z1}i})+  R_{Z}(\theta_{\text{Z1}i})R_{Y}(\theta_{\text{Y}i})(\imath Z)R_{Z}(\pi v-\theta_{\text{Z1}i})).
\end{align}
Evaluating for $\boldsymbol{\theta}_\I$
\begin{equation}\label{dtheta1}
\begin{aligned}
 \partial_{\text{Z1}i} V_{\text{hea}}(\boldsymbol{\theta})|_{\boldsymbol{\theta}=\boldsymbol{\theta}_\I} &=\imath Z-\imath Z
 =0. 
\end{aligned}
\end{equation}
From \eqref{dtheta1} and \eqref{distance_derivative},  $\partial_{\text{Z1}i}d(X^{a_1...a_k}_b,V_{\text{hea}}(\boldsymbol{\theta}))|_{\boldsymbol{\theta}=\boldsymbol{\theta}_\I}=0$. 
As for the second derivative
\begin{align}
    \partial^2_{\text{Z1}i} V_{\text{hea}}(\boldsymbol{\theta})|_{\boldsymbol{\theta}=\boldsymbol{\theta}_\I}=&(\imath Z)^2-(\imath Z)^2-(\imath Z)^2+(\imath Z)^2=0,\\
    \partial^2_{\text{Z1}i}d(X^{a_1...a_k}_b,V_{\text{hea}}(\boldsymbol{\theta}))|_{\boldsymbol{\theta}=\boldsymbol{\theta}_\I}=&0,
\end{align}
which indicates this is not just a extrema point but a flat valley.

The derivative of $V_{\text{hea}}(\boldsymbol{\theta})$ with respect to an $R_Y$ parameter is
\begin{equation}\label{devz}
    \begin{aligned}
 \partial_{\text{Y}i} V_{\text{hea}}(\boldsymbol{\theta})= \cdots R_{Z}(\theta_{\text{Z1}i})(-\imath Y)R_ Y(\theta_{\text{Y}i})R_Z(\pi v-\theta_{\text{Z1}i}).
\end{aligned}
\end{equation}
Evaluating for $\boldsymbol{\theta}_\I$
\begin{equation}
    \begin{aligned}
 \partial_{{\text{Y}i}} V_{\text{hea}}(\boldsymbol{\theta})|_{\boldsymbol{\theta}=\boldsymbol{\theta}_\I}&=e^{\imath\alpha(\boldsymbol{\theta}_\I)} R_Z(\theta_{\text{Z1}i})(-\imath Y) R_Z(-\theta_{\text{Z1}i})\\
 &= e^{\imath\alpha(\boldsymbol{\theta}_\I)}\begin{pmatrix}
    0 &  -e^{\imath(-2\theta_{\text{Z1}i})} \\
   e^{\imath(2\theta_{\text{Z1}i})}  &  0 \\
\end{pmatrix}.
\end{aligned}
\label{Rytheta2_derivative}
\end{equation}
From \eqref{Rytheta2_derivative} and \eqref{distance_derivative}, $\partial_{\text{Y}i}d(X^{a_1...a_k}_b,V_{\text{hea}}(\boldsymbol{\theta}))|_{\boldsymbol{\theta}=\boldsymbol{\theta}_\I}=0$. The second derivative is
\begin{equation}
    \begin{aligned}
 \partial^2_{\text{Y}i} V_{\text{hea}}(\boldsymbol{\theta})=& \cdots R_Z(\theta_{\text{Z1}i})(-\imath Y)^2R_Y(\theta_{\text{Y}i})R_Z(\pi v-\theta_{\text{Z1}i})\\
\partial^2_{\text{Y}i} V_{\text{hea}}(\boldsymbol{\theta})|_{\boldsymbol{\theta}=\boldsymbol{\theta}_\I} 
 = &-e^{\imath\alpha(\boldsymbol{\theta}_\I)} \I.
\end{aligned}
\label{Rytheta2secondderiv}
\end{equation}
Substituting \eqref{Rytheta2secondderiv} in \eqref{distance_second_derivative} gives $\partial^2_{\text{Y}i}d(X^{a_1...a_k}_b,V_{\text{hea}}(\boldsymbol{\theta}))|_{\boldsymbol{\theta}=\boldsymbol{\theta}_\I}>0$.

The controlled $R_Y$ can be written as 
\[CR_Y(\theta_{\text{D}i})=(\I\otimes R_Y(\theta_{\text{D}i}/2))\CN (\I\otimes R_Y(-\theta_{\text{D}i}/2))\CN.\]
The partial derivative of $V_{\text{hea}}(\boldsymbol{\theta})$ with respect to a controlled gate is
\begin{equation}\label{devcr}
\begin{aligned}
    \partial_{\text{D}i}V_{\text{hea}}(\boldsymbol{\theta})=&\bigg(\left(\I\otimes{\frac{-\imath Y}{2}} R_Y(\theta_{\text{D}i}/2)\right){\CN}(\I\otimes R_Y(-\theta_{\text{D}i}/2)){\CN}\\
    &+(\I\otimes R_Y(\theta_{\text{D}i}/2))\CN\left(\I\otimes{ \frac{\imath Y}{2}} R_Y(-\theta_{\text{D}i}/2)\right)\CN\bigg)\cdots.
    \end{aligned}
\end{equation}
Evaluating for $\boldsymbol{\theta}_\I$
\begin{equation}
\begin{aligned}
    \partial_{\text{D}i}V_{\text{hea}}(\boldsymbol{\theta})|_{\boldsymbol{\theta}=\boldsymbol{\theta}_\I}= &e^{\imath\alpha(\boldsymbol{\theta}_\I)}\left(\I\otimes{ \frac{-\imath Y}{2}} \right)\CN(\I\otimes\I )\CN\\
    &+e^{\imath\alpha(\boldsymbol{\theta}_\I)}(\I\otimes \I )\CN\left(\I\otimes{\frac{\imath Y}{2}} \right)\CN\\
    = &e^{\imath\alpha(\boldsymbol{\theta}_\I)}\left[\left(\I\otimes{\frac{-\imath Y}{2}} \right)+\CN\left(\I\otimes{\frac{\imath Y}{2}} \right)\CN\right]\\
    =&e^{\imath\alpha(\boldsymbol{\theta}_\I)}\left[{\frac{\imath}{2}}(-(Y\oplus Y)+\CN (Y\oplus Y))\CN)\right]\\
    =& e^{\imath\alpha(\boldsymbol{\theta}_\I)}{\frac{\imath}{2}}(-(Y\oplus Y)+(Y\oplus (-Y)))\\
   = & e^{\imath\alpha(\boldsymbol{\theta}_\I)} \boldsymbol{\underline{0}}_2\oplus(-\imath Y).
\end{aligned}
\label{CRy_derivative}
\end{equation}
From \eqref{CRy_derivative} and \eqref{distance_derivative} $\partial_{\text{D}i}d(X^{a_1...a_k}_b,V_{\text{hea}}(\boldsymbol{\theta}))|_{\boldsymbol{\theta}=\boldsymbol{\theta}_\I}=0$. The second derivative of $CR_Y(\theta)$ is
\begin{equation}
\begin{aligned}
\partial^2_{\text{D}i}V_{\text{hea}}(\boldsymbol{\theta})=&\left(\I\otimes\left({\frac{-\imath Y}{2}}\right)^2 R_Y(\theta_{\text{D}i}/2)\right)\CN(\I\otimes R_Y(-\theta_{\text{D}i}/2))\CN\\
&+\left(\I\otimes{\frac{-\imath Y}{ 2}} R_Y(\theta_{\text{D}i}/2)\right)\CN\left(\I\otimes{\frac{\imath Y}{ 2}} R_Y(-\theta_{\text{D}i}/2)\right)\CN\\
&+(\I\otimes{\frac{-\imath Y}{ 2}} R_Y(\theta_{\text{D}i}/2))\CN\left(\I\otimes{\frac{\imath Y}{ 2}} R_Y(-\theta_{\text{D}i}/2)\right)\CN\\
&+(\I\otimes R_Y(\theta_{\text{D}i}/2))\CN\left(\I\otimes\left({\frac{\imath Y}{2}}\right)^2 R_Y(-\theta_{\text{D}i}/2)\right)\CN,
\end{aligned}
\end{equation}

\begin{equation}
\begin{aligned}
  \partial^2_{\text{D}i}V_{\text{hea}}(\boldsymbol{\theta})|_{\boldsymbol{\theta}=\boldsymbol{\theta}_\I}=&e^{\imath\alpha(\boldsymbol{\theta}_\I)}\Bigg[ -{\frac{1}{4}}\left(\I\otimes\I \right)\CN(\I\otimes\I )\CN
  +\left(\I\otimes{\frac{-\imath Y}{2}} \right)\CN\left(\I\otimes{\frac{\imath Y}{2}} \right)\CN\\
 & +\left(\I\otimes{\frac{-\imath Y}{2}} \right)\CN\left(\I\otimes{\frac{\imath Y}{2}} \right)\CN
  -{\frac{1}{4}}\left(\I\otimes\I \right)\CN(\I\otimes\I )\CN\Bigg]\\
  = & e^{\imath\alpha(\boldsymbol{\theta}_\I)}\left[\frac{1}{4} (\I_2\oplus(-\I_2)) +\frac{1}{4} (\I_2\oplus(-\I_2)) -\frac{1}{4} ~\I_4-\frac{1}{4} ~\I_4 \right]\\
 = & e^{\imath\alpha(\boldsymbol{\theta}_\I)} \boldsymbol{\underline{0}}_2\oplus(-\I_2). 
\end{aligned}
\label{CRy_second_derivative}
\end{equation}
From \eqref{CRy_second_derivative}, \eqref{distance_second_derivative}, $\partial^2_{\text{D}i}d(X^{a_1...a_k}_b,V_{\text{hea}}(\boldsymbol{\theta}))|_{\boldsymbol{\theta}=\boldsymbol{\theta}_\I}>0$.


Secondly, we compute the derivatives for the second set of constraints on the parameters \eqref{heavarietyeqs}--(b).
The derivative of  $V_{\text{hea}}(\boldsymbol{\theta})$ with respect to a pair of $R_Z$ parameters we have
\begin{equation}
\begin{aligned}
\partial_{\text{Z1}i} V_{\text{hea}}(\boldsymbol{\theta})=&\cdots((-\imath Z) R_{Z}(\theta_{\text{Z1}i})R_{Y}(\theta_{\text{Y}i})R_{Z}(\pi(1/2+ v)-\theta_{\text{Z1}i})Z\\
&+ R_{Z}(\theta_{\text{Z1}i})R_{Y}(\theta_{\text{Y}i})(\imath Z)R_{Z}(\pi(1/2+ v)-\theta_{\text{Z1}i}Z)).
 \end{aligned}
\end{equation}
Evaluating for $\boldsymbol{\theta}_\I$
\begin{equation}\label{2setdtheta1}
\begin{aligned}
 \partial_{\text{Z1}i} V_{\text{hea}}(\boldsymbol{\theta})|_{\boldsymbol{\theta}=\boldsymbol{\theta}_\I} &=\imath Z-\imath Z
 =0. 
\end{aligned}
\end{equation}
From \eqref{2setdtheta1} and \eqref{distance_derivative}, $\partial_{\text{Z1}i}d(X^{a_1...a_k}_b,V_{\text{hea}}(\boldsymbol{\theta}))|_{\boldsymbol{\theta}=\boldsymbol{\theta}_\I}=0$.
The second derivative is 
\begin{equation}
\begin{aligned}
     \partial^2_{\text{Z1}i} V_{\text{hea}}(\boldsymbol{\theta})|_{\boldsymbol{\theta}=\boldsymbol{\theta}_\I}=(\imath Z)^2-(\imath Z)^2-(\imath Z)^2+(\imath Z)^2&=0.\\
     \partial^2_{\text{Z1}i}d(X^{a_1...a_k}_b,V_{\text{hea}}(\boldsymbol{\theta}))|_{\boldsymbol{\theta}=\boldsymbol{\theta}_\I}&=0
\end{aligned}
\end{equation}

Following a similar process as in \eqref{devz} and \eqref{devcr} one can calculate the partial derivatives with respect to the $R_Y$ and $CR_Y$ parameters appearing in \eqref{heavarietyeqs}--(b). The gradients with respect to the $R_Y$ parameters

\begin{equation}\label{eqset2}
\begin{aligned}
 \partial_{{\text{Y}i}}  V_{\text{hea}}(\boldsymbol{\theta})|_{\boldsymbol{\theta}=\boldsymbol{\theta}_\I}& =e^{\imath\alpha(\boldsymbol{\theta}_\I)} \begin{pmatrix}
    0 &  -e^{\imath(-2\theta_{\text{Z1}i})} \\
   e^{\imath(2\theta_{\text{Z1}i})}  &  0 \\
\end{pmatrix},    
\\
\partial_{Yi}d(X^{a_1...a_k}_b,V_{\text{hea}}(\boldsymbol{\theta}))|_{\boldsymbol{\theta}=\boldsymbol{\theta}_\I}&=0,    
\\
   \partial^2_{\text{Y}i} V_{\text{hea}}(\boldsymbol{\theta})|_{\boldsymbol{\theta}=\boldsymbol{\theta}_\I} &=-e^{\imath\alpha(\boldsymbol{\theta}_\I)}\I,\\
\partial^2_{Yi}d(X^{a_1...a_k}_b,V_{\text{hea}}(\boldsymbol{\theta}))|_{\boldsymbol{\theta}=\boldsymbol{\theta}_\I}&\geq 0.
    \end{aligned}
\end{equation}
The gradients with respect to the $CR_Y$ parameters
\begin{equation}
\begin{aligned}
 \partial_{\text{D}i}V_{\text{hea}}(\boldsymbol{\theta})|_{\boldsymbol{\theta}=\boldsymbol{\theta}_\I}&=-\imath e^{\imath\alpha(\boldsymbol{\theta}_\I)}( \boldsymbol{\underline{0}}_2\oplus Y),\\
\partial_{Di}d(X^{a_1...a_k}_b,V_{\text{hea}}(\boldsymbol{\theta}))|_{\boldsymbol{\theta}=\boldsymbol{\theta}_\I}&=0,  \\
\partial^2_{\text{D}i}V_{\text{hea}}(\boldsymbol{\theta})|_{\boldsymbol{\theta}=\boldsymbol{\theta}_\I}&=-e^{\imath\alpha(\boldsymbol{\theta}_\I)}(\boldsymbol{\underline{0}}_2\oplus\I_2),\\
\partial^2_{Di}d(X^{a_1...a_k}_b,V_{\text{hea}}(\boldsymbol{\theta}))|_{\boldsymbol{\theta}=\boldsymbol{\theta}_\I}&\geq0.
    \end{aligned}
\end{equation}
\end{proof}


\begin{figure}[ht]
\[   \Qcircuit @C=1.5em @R= 01.5em {
& \gate{R_X(\omega_1)} & \multigate{1}{R_{ZZ}}    & \gate{R_Z(\omega_4)}   & \qw& \qw & \qw & \qw   \\
& \gate{R_X(\omega_2)} & \ghost{R_{ZZ}}    &    \gate{R_Z(\omega_5)} & \gate{R_X(\omega_6)} &\multigate{1}{R_{ZZ}}  &\gate{R_Z(\omega_9)}&\qw\\
 &\qw &\qw & \qw &  \gate{R_X(\omega_7)} &\ghost{R_{ZZ}} &\gate{R_Z(\omega_{10})}&\qw
\gategroup{1}{2}{2}{4}{.75em}{--} \gategroup{2}{5}{3}{7}{.75em}{--}
}\]
    \caption{Our implementation of a single layer of the checkerboard ansatz for 3 qubits, $R_{ZZ}(\omega)=e^{-\imath\omega Z\otimes Z}$.}
    \label{fig:3q-che}
\end{figure}
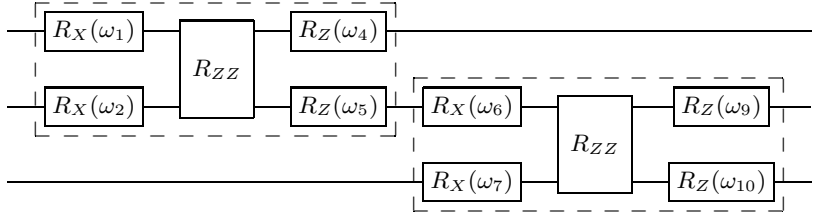

\begin{theorem} [Extrema of the checkerboard ansatz] For a single layer of the checkerboard ansatz, $d$ takes extrema for the roots of  
\begin{equation}
    V_\text{ch}(\boldsymbol{\omega}_\I)-e^{\imath\alpha(\boldsymbol{\omega}_\I)}\I=0.
    \label{cheidentitycondition}
\end{equation}

\end{theorem}

\begin{proof}
We prove $V_\text{ch}(\boldsymbol{\omega}_\I)$ are extrema by showing
\begin{equation}
\begin{aligned}
    \partial_j d(V_{ch}(\boldsymbol{\omega}),X^{a_1...a_k}_b)|_{\boldsymbol{\omega}=\boldsymbol{\omega}_\I}=0,\\
    \partial^2_j d(V_{ch}(\boldsymbol{\omega}),X^{a_1...a_k}_b)|_{\boldsymbol{\omega}=\boldsymbol{\omega}_\I}\geq0.
\end{aligned}
\end{equation}
Following the gate order in figure \ref{fig:3q-che}, the derivative of $V_\text{ch}(\boldsymbol{\omega})$ with respect to the first $R_X$ parameter is
\begin{equation}
\begin{aligned}
    \partial_{\text{X1}i}V_\text{ch}(\boldsymbol{\omega})=&R_Z(\omega_{\text{Z2}i})R_{ZZ}(\omega_{\text{ZZ2}i})R_X(\omega_{\text{X2}i})R_Z(\omega_{\text{Z1}i})R_{ZZ}(\omega_{\text{ZZ1}i})R_X(\omega_{\text{X1}i})(-\imath X)\\
    =&V_\text{ch}(\boldsymbol{\omega})(-\imath X).\\
    \end{aligned}
    \end{equation}
    Evaluating for $\boldsymbol{\omega}_\I$,
    \begin{equation}
\begin{aligned}
     \partial_{\text{X1}i}V_\text{ch}(\boldsymbol{\omega})|_{\boldsymbol{\omega}=\boldsymbol{\omega}_\I}=&e^{\imath\alpha(\boldsymbol{\omega}_\I)}(-\imath X),\\
     \partial_{\text{X1}i}d(X^{a_1...a_k}_b,V_\text{ch}(\boldsymbol{\omega}))|_{\boldsymbol{\omega}=\boldsymbol{\omega}_\I}=&0.
\end{aligned}
\end{equation}

The derivative of $V_\text{ch}(\boldsymbol{\omega})$ with respect to the first $R_Z$ parameter is
\begin{equation}
    \begin{aligned}
     \partial_{\text{Z1}i}V_\text{ch}(\boldsymbol{\omega})=&R_Z(\omega_{\text{Z2}i})R_{ZZ}(\omega_{\text{ZZ2}i})R_X(\omega_{\text{X2}i})R_Z(\omega_{\text{Z1}i})R_{ZZ}(\omega_{\text{ZZ1}i})(-\imath Z)R_X(\omega_{\text{X1}i}).\\
    \end{aligned}
\end{equation}
From \eqref{cheidentitycondition} and evaluating for $\boldsymbol{\omega}_\I$, 
\begin{equation}
    \begin{aligned}
     \partial_{\text{Z1}i}V_\text{ch}(\boldsymbol{\omega})|_{\boldsymbol{\omega}=\boldsymbol{\omega}_\I}=&e^{\imath\alpha(\boldsymbol{\omega}_\I)}R_X(\omega_{\text{X1}i})^\dagger(-\imath Z)R_X(\omega_{\text{X1}i})\\
     =&e^{\imath\alpha(\boldsymbol{\omega}_\I)}\begin{pmatrix}
    -\imath(\cos^2(\omega_{\text{X1}i})-\sin^2(\omega_{\text{X1}i})) &  -2\sin(\omega_{\text{X1}i})\cos(\omega_{\text{X1}i}) \\
   2\sin(\omega_{\text{X1}i})\cos(\omega_{\text{X1}i})  &  \imath(\cos^2(\omega_{\text{X1}i})-\sin^2(\omega_{\text{X1}i})) \\
\end{pmatrix}.
\label{Z1mattrix}
    \end{aligned}
\end{equation}
The matrix elements in opposite corners of \eqref{Z1mattrix} are the same but with opposite signs, which makes them cancel when evaluating the distance derivative
\begin{equation}
    \partial_{\text{Z1}i}d(X^{a_1...a_k}_b,V_\text{ch}(\boldsymbol{\omega}))|_{\boldsymbol{\omega}=\boldsymbol{\omega}_\I}=0.
\end{equation}

The derivative of $V_\text{ch}(\boldsymbol{\omega})$ with respect to the second $R_Z$ parameter is
\begin{equation}
    \begin{aligned}
     \partial_{\text{Z2}i}V_\text{ch}(\boldsymbol{\omega})=&(-\imath Z)R_Z(\omega_{\text{Z2}i})R_{ZZ}(\omega_{\text{ZZ2}i})R_X(\omega_{\text{X2}i})R_Z(\omega_{\text{Z1}i})R_{ZZ}(\omega_{\text{ZZ1}i})R_X(\omega_{\text{X1}i})\\
     =&(-\imath Z)V_\text{ch}(\boldsymbol{\omega}).
    \end{aligned}
\end{equation}
Evaluating for $\boldsymbol{\omega}_\I$,  
\begin{equation}
\begin{aligned}
     \partial_{\text{Z2}i}V_\text{ch}(\boldsymbol{\omega})|_{\boldsymbol{\omega}=\boldsymbol{\omega}_\I}=&e^{\imath\alpha(\boldsymbol{\omega}_\I)}(-\imath Z),\\
     \partial_{\text{Z2}i}d(X^{a_1...a_k}_b,V_\text{ch}(\boldsymbol{\omega}))|_{\boldsymbol{\omega}=\boldsymbol{\omega}_\I}=&0.
\end{aligned}
\end{equation}

The derivative of $V_\text{ch}(\boldsymbol{\omega})$ with respect to the second $R_X$ parameter is
\begin{equation}
    \begin{aligned}
     \partial_{\text{X2}i}V_\text{ch}(\boldsymbol{\omega})=&R_Z(\omega_{\text{Z2}i})R_{ZZ}(\omega_{\text{ZZ2}i})(-\imath X)R_X(\omega_{\text{X2}i})R_Z(\omega_{\text{Z1}i})R_{ZZ}(\omega_{\text{ZZ1}i})R_X(\omega_{\text{X1}i}).\\
    \end{aligned}
\end{equation}
From \eqref{cheidentitycondition} and evaluating for $\boldsymbol{\omega}_\I$, 
\begin{equation}
    \begin{aligned}
     \partial_{\text{X2}i}V_\text{ch}(\boldsymbol{\omega})|_{\boldsymbol{\omega}=\boldsymbol{\omega}_\I}=&
     e^{\imath\alpha(\boldsymbol{\omega}_\I)}R_Z(\omega_{\text{Z2}i})R_{ZZ}(\omega_{\text{ZZ2}i})(-\imath X)R_{ZZ}(\omega_{\text{ZZ2}i})^\dagger R_Z(\omega_{\text{Z2}i})^\dagger.\\
\end{aligned}
\end{equation}
$R_Z(\omega_{\text{Z2}i})R_{ZZ}(\omega_{\text{ZZ2}i})$ is diagonal with elements of the form $e^{\imath \beta_m}$. where $\beta_m\in \mathbb{Z}$, $m\in\{1,2,3,4\}$, which makes 
\begin{equation}
    \begin{aligned}
     \partial_{\text{X2}i}V_\text{ch}(\boldsymbol{\omega})|_{\boldsymbol{\omega}=\boldsymbol{\omega}_\I}=&-\imath e^{\imath\alpha(\boldsymbol{\omega}_\I)}
     \begin{pmatrix}
    0 & e^{\imath (\beta_1-\beta_2)}  & 0 & 0\\
    e^{-\imath (\beta_1-\beta_2)} & 0 & 0 & 0\\
    0 & 0 & 0 & e^{\imath (\beta_3-\beta_4)} \\
    0 & 0 & e^{-\imath (\beta_3-\beta_4)}  & 0
\end{pmatrix}.
     \end{aligned}
     \label{X2der}
     \end{equation}
When evaluating the distance derivative, the Hilbert-Schmidt product between \eqref{X2der} and $X^{a_1...a_k}_b$ will be a purely complex number, since \eqref{distance_derivative} takes only the real part    
\begin{equation}
    \partial_{\text{X2}i}d(X^{a_1...a_k}_b,V_\text{ch}(\boldsymbol{\omega}))|_{\boldsymbol{\omega}=\boldsymbol{\omega}_\I}=0.
\end{equation}

The derivative with respect to the second $R_{ZZ}$ parameter is
\begin{equation}
    \begin{aligned}
     \partial_{\text{ZZ2}i}V_\text{ch}(\boldsymbol{\omega})=&(-\imath Z\otimes Z)V_\text{ch}(\boldsymbol{\omega}).
    \end{aligned}
\end{equation}
Evaluating for $\boldsymbol{\omega}_\I$  
\begin{equation}
\begin{aligned}
     \partial_{\text{ZZ2}i}V_\text{ch}(\boldsymbol{\omega})|_{\boldsymbol{\omega}=\boldsymbol{\omega}_\I}=&e^{\imath\alpha(\boldsymbol{\omega}_\I)}(-\imath Z\otimes Z),\\
     \partial_{\text{ZZ2}i}d(X^{a_1...a_k}_b,V_\text{ch}(\boldsymbol{\omega}))|_{\boldsymbol{\omega}=\boldsymbol{\omega}_\I}=&0.
\end{aligned}
\end{equation}

The derivative of $V_\text{ch}(\boldsymbol{\omega})$ with respect to the first $R_{ZZ}$ parameter is
\begin{equation}
    \begin{aligned}
     \partial_{\text{ZZ1}i}V_\text{ch}(\boldsymbol{\omega})=&R_Z(\omega_{\text{Z2}i})R_{ZZ}(\omega_{\text{ZZ2}i})R_X(\omega_{\text{X2}i})R_Z(\omega_{\text{Z1}i})R_{ZZ}(\omega_{\text{ZZ1}i})(-\imath Z\otimes Z)R_X(\omega_{\text{X1}i}).\\
    \end{aligned}
\end{equation}
From \eqref{cheidentitycondition} and evaluating for $\boldsymbol{\omega}_\I$, 
\begin{equation}
    \begin{aligned}
     \partial_{\text{Z1}i}V_\text{ch}(\boldsymbol{\omega})|_{\boldsymbol{\omega}=\boldsymbol{\omega}_\I}=&e^{\imath\alpha(\boldsymbol{\omega}_\I)}(R_X(\omega_{\text{X1}i})\otimes R_X(\omega_{\text{X1}i+1}))^\dagger(-\imath Z\otimes Z)(R_X(\omega_{\text{X1}i})\otimes R_X(\omega_{\text{X1}i+1}))\\
     =&e^{\imath\alpha(\boldsymbol{\omega}_\I)}\begin{pmatrix}
    -\imath(\cos^2(\omega_{\text{X1}i})-\sin^2(\omega_{\text{X1}i})) &  -2\sin(\omega_{\text{X1}i})\cos(\omega_{\text{X1}i}) \\
   2\sin(\omega_{\text{X1}i})\cos(\omega_{\text{X1}i})  &  \imath(\cos^2(\omega_{\text{X1}i})-\sin^2(\omega_{\text{X1}i})) \\
\end{pmatrix}\\
&\otimes \begin{pmatrix}
    -\imath(\cos^2(\omega_{\text{X1}i+1})-\sin^2(\omega_{\text{X1}i+1})) &  -2\sin(\omega_{\text{X1}i+1})\cos(\omega_{\text{X1}i+1}) \\
   2\sin(\omega_{\text{X1}i+1})\cos(\omega_{\text{X1}i+1})  &  \imath(\cos^2(\omega_{\text{X1}i+1})-\sin^2(\omega_{\text{X1}i+1})) \\
   \end{pmatrix}.\\
\label{zz1matrix}
    \end{aligned}
\end{equation}
Similar to \eqref{Z1mattrix}, the elements in opposite corners of the matrices in \eqref{zz1matrix} are the same but with opposite signs, which makes them cancel when evaluating the distance derivative
\begin{equation}
    \partial_{\text{Z1}i}d(X^{a_1...a_k}_b,V_\text{ch}(\boldsymbol{\omega}))|_{\boldsymbol{\omega}=\boldsymbol{\omega}_\I}=0.
\end{equation}

The second derivative of $V_\text{ch}(\boldsymbol{\omega})$ with respect to any parameter is
\begin{equation}
    \begin{aligned}
     \partial_{i}^2V_\text{ch}(\boldsymbol{\omega})
     =&-V_\text{ch}(\boldsymbol{\omega}),\\
     \partial_{i}^2 d(X^{a_1...a_k}_b,V_\text{ch}(\boldsymbol{\omega}))|_{\boldsymbol{\omega}=\boldsymbol{\omega}_\I}>&0.
    \end{aligned}
\end{equation}

\end{proof}

\end{document}